\pdfoutput=1
\documentclass[11pt]{article}
\usepackage{times}
\usepackage{helvet}
\usepackage{courier}
\usepackage{fullpage}
\usepackage[norelsize,linesnumbered,ruled,vlined,boxed]{algorithm2e}
\usepackage{bbm}

\usepackage{url}
\usepackage{float}
\restylefloat{table}
\restylefloat{figure}
\usepackage{graphicx}

\usepackage{amsmath,amssymb,amstext,amsfonts} 
\usepackage{enumerate}
\usepackage{subfigure}
\usepackage{color}
\usepackage{latexsym}
\usepackage{comment}

\newcount\Comments
\definecolor{darkgreen}{rgb}{0,0.7,0}
\newcommand{\kibitz}[2]{\ifnum\Comments=1\textcolor{#1}{#2}\fi}

\renewcommand{\vec}[1]{\mathbf{#1}}

\newtheorem{example}{Example}
\newtheorem{theorem}{Theorem}
\newenvironment{proof}{\paragraph{Proof:}}{\hfill$\square$}

        {\medskip}

\usepackage{enumerate}
\usepackage[linesnumbered,ruled,vlined,boxed]{algorithm2e}
\usepackage{wrapfig}

\usepackage{subfigure}
\usepackage{color}

\title{To Give or Not to Give: Fair Division for Single Minded Valuations\thanks{The authors aknowledge support from
the Danish National
Research Foundation and The National Science Foundation of China
(under the grant
61361136003) for the Sino-Danish Center for the Theory of Interactive Computation and
from the Center for Research in Foundations of Electronic Markets (CFEM), supported by
the Danish Strategic Research Council.
A part of this work was done while Simina Br\^{a}nzei and Ruta Mehta were visiting the Simons Institute for the Theory of Computing.
Simina Br\^{a}nzei was also supported by ISF grant 1435/14 administered by the Israeli Academy of Sciences and Israel-USA Bi-national Science Foundation (BSF) grant  2014389, as well as the I-CORE Program of the Planning and Budgeting Committee and The Israel Science Foundation.
}}

\author{Simina Br\^anzei\footnote{
Hebrew University of Jerusalem, Israel. E-mail: {\tt simina.branzei@gmail.com}}
\and
Yuezhou Lv\footnote{
Tsinghua University, China. E-mail: {\tt totolv@126.com}}\\
\and
Ruta Mehta\footnote{
University of Illinois Urbana-Champaign, USA. E-mail: {\tt rutameht@illinois.edu}}
}
\date{}

\begin{document}

\maketitle


\begin{abstract}
Single minded agents have strict preferences, in which a bundle is acceptable only if it meets a certain demand. Such preferences arise
naturally in scenarios such as allocating computational resources among users, where the goal is to fairly serve as many
requests as possible. In this paper we study the fair division problem for such agents, which is harder to handle due to discontinuity and complementarities of
the preferences.

Our solution concept---the \emph{competitive allocation from equal incomes} (CAEI)---is inspired from market equilibria and implements
fair outcomes through a pricing mechanism. We study the existence and computation of CAEI for multiple divisible goods, cake cutting, and multiple
discrete goods.
For the first two scenarios we show that existence of CAEI solutions is guaranteed, while for the third we give a succinct characterization of instances that admit this solution; then we give an efficient algorithm to find one in all three cases. Maximizing social welfare turns out to be NP-hard in general, however we obtain efficient algorithms for (i) divisible and discrete goods when the number of different \emph{types} of players is a constant, (ii) cake cutting with contiguous demands, for which we establish an interesting connection with interval scheduling, and (iii) cake cutting with a constant number of players with arbitrary demands.

Our solution is useful more generally, when the players have a target set of desired goods, and very small positive values for any bundle not containing their target set.
\end{abstract}

\newpage

\section{Introduction} \label{sec:intro}

The question of dividing scarce resources among multiple participants in a way that is fair 
has remained a pressing question that intrigued humans for a long time, 
be it for dividing land among citizens\footnote{Divide-and-choose is mentioned already in the Bible, in the Book of Genesis (chapter
13)} or allocating organizational resources to its members. 
The formal study of {\em fair division} as we know it started during World War II with works of Steinhaus, Knaster and Banach
\cite{Steinhaus51}. Much of the earlier work focused on one heterogeneous divisible good, i.e. the cake cutting
problem~\cite{BT96,RW98}, while other models (such as the problem of allocating multiple divisible or discrete goods) and notions of fairness were studied later in a growing body of
literature, which includes multiple academic books~\cite{BT96,RW98,Moulin04}. Fair division has been recently studied in the computer science community, as problems in resource
allocation and fair division in particular are arguably relevant in scenarios such as manufacturing and
scheduling, airport traffic, and industrial procurement~\cite{CDEL+06,Pro13}. 

Prominent examples of fair division models that surfaced recently in real scenarios include problem of allocating computational
resources (e.g. CPU, memory, bandwith) among the users of a system~\cite{GZHB+11}, or the problem of allocating university courses to
students in a way that is fair and efficient. The latter motivated the introduction of a notion of fairness known as
A-CEEI~\cite{Budish11}, which approximates the well known ideal notion of fairness from economics, the \emph{competitive equilibrium
from equal incomes}~\cite{Foley67,Varian74}. Othman \emph{et al.} [\cite{OPR14}] studied the complexity of A-CEEI with pessimistic
conclusions for general preferences. Nevertheless, the A-CEEI solution is used to allocate courses to students at the Wharton Business
School at the University of Pennsylvania.

The largest body of the literature on fair division models for multiple divisible and indivisible goods, as well as cake cutting, 
focuses on additive valuations, which capture \emph{perfect substitutes}, i.e. goods that can replace each other in consumption,
such as Pepsi and Coca-Cola. However, many real allocation problems have complementarities in the preferences to various degrees. For
example, if the user of a computer system wants to run a specific task, they may need $2$ units of CPU and $5$ units of RAM. It is not
useful if the user gets only $1$ unit of CPU but $6$ or more units of RAM -- they simply cannot run the task because CPU is a
bottleneck resource in this case.

An important scenario with complements is that of \emph{single minded} valuations, where the agent values a particular bundle and
nothing less, and anything extra does not add to the value. 
Such valuations arise naturally, for example assembling a bike requires a set of parts, or an agent's computational task can finish if and
only if it is allocated a required bundle of resource. 
The latter example represents basis for more complex models that take into
account the dynamics over time. Due to their immense applicability, there is an extensive body of literature on such single minded
valuations in areas such as auctions (see, e.g.,~\cite{BKV05,NRTV07,L07,RPIJ13}), however work on fair division with such valuations is
largely missing. 

\subsection{Our Contribution}

We study fair division among single minded agents for three main scenarios, namely multiple divisible goods, cake
cutting, and discrete goods. Our main solution concept---the
``competitive allocation from equal incomes'' (CAEI)---is inspired from market equilibria and implements fair outcomes through a
pricing mechanism. The CAEI solution can be seen as a relaxation of the standard competitive equilibrium from equal incomes and is a
more intuitive notion for discrete resources than the latter. However many of our results, most notably for multiple
divisible goods, carry over to the competitive equilibrium from equal incomes solution. The CAEI solution concept is sandwiched
between the competitive equilibrium from equal incomes and envy-freeness, and it prunes (via the prices), the least desirable envy-free
allocations. Nevertheless, CAEI  exists if and only if envy-free allocations exist in all the fair division models we consider.

For multiple divisible goods and cake cutting we show existence of CAEI, while for discrete goods we give a succinct characterization of
instances that admit this solution; then we give an efficient algorithm to find one in all cases. Maximizing social welfare turns
out to be NP-hard in general, however we obtain efficient algorithms for a number of settings, (i) divisible and discrete goods when
the number of different \emph{types} of agents is a constant, (ii) cake cutting with contiguous demands, for which we establish an interesting connection with interval scheduling, and (iii) cake cutting with
constant number of agents.

Our results also carry over to valuation models where the agents have a desired target set (e.g. the set of parts of a bike),
while having very small positive value ($\epsilon$) for any bundle that does not contain their desired set. For such scenarios (which are compatible with 
experimental evidence that people like to accumulate ``stuff'', even if unneeded), the solution computed is an $\epsilon$-CAEI.

\section{Background} \label{sec:model}

We begin by presenting the model and solution concept for the most general setting. Let $N = \{1, \ldots, n\}$ be a set of
agents and $\mathcal{R}$ a set of resources, where $\mathcal{R}$ is a compact subset of a Euclidean space.  Each agent $i$ is
equipped with a valuation function $V_i$ over the resources, such that for each subset $S \subseteq \mathcal{R}$, $V_i(S)$ represents
the valuation of agent $i$ for bundle $S$. The goal is to allocate the resources among the agents in a way that is fair.

Agent $i$ is said to be \emph{single minded} if there exists a bundle $D_i \subseteq \mathcal{R}$ that $i$ cannot be happy without;
that is, for any other bundle $S \subseteq \mathcal{R}$ we have that $V_i(S) = 1$ if $D_i \subset S$ and $V_i(S) = 0$ otherwise.

\vspace{2mm}
An \emph{allocation} $\vec{x} = (\vec{x}_1, \ldots, \vec{x}_n)$ is a partition of the set $\mathcal{R}$ of resources such that
$\vec{x}_i \subseteq \mathcal{R}$ is the bundle received by agent $i$, the bundles are not intersecting and add up to the whole space.

Our solution concept---the ``competitive allocation from equal incomes'' (CAEI)---is inspired from market equilibria and implements
fair outcomes through a pricing mechanism. More formally, each agent is given an artificial unit of currency 
by the center.  It can be used to acquire goods, and has no intrinsic value. The agent wants to spend its budget to acquire a bundle
of items that maximizes its utility. A CAEI outcome is defined as a tuple $\langle p, \vec{x}\rangle$, where $\vec{x}$ is an allocation
and $p: \mathcal{R} \rightarrow \mathbb{R}_{+}$ is an integrable, non-negative price density function, such that for each subset $S
\subseteq \mathcal{R}$, $$p(S) = \int_{x \in S} p(x) dx.$$

Formally, a tuple $\langle p, \vec{x} \rangle$ is a CAEI solution if and only if:
\begin{itemize}
\item For all $i \in N$, $\vec{x}_i$ maximizes agent $i$'s utility given prices $p$ and its unit budget.

\item All the resources are allocated: $\bigcup_{i=1}^{n} \vec{x}_{i} = \mathcal{R}$.
\item No agent overspends: $p(\vec{x}_i) \leq 1$ for all $i \in N$.
\end{itemize}

The difference between the CAEI solution and the classical notion of a competitive equilibrium from equal incomes (CEEI), is that under CEEI, the agents are additionally required to spend all of their budget. Since we are in a
fair division setting, neither the center nor the agents have value for the units of currency, thus the requirement of spending all
the budgets can be relaxed naturally. Moreover, in several of the scenarios we study, the CEEI solution is not as tractable due to discontinuity of the valuations and/or discrete goods. 
Crucially, every CAEI allocation is envy-free, since if an agent envied another agent's bundle it could just buy it
instead, as the endowments are equal; the prices provide a mechanism for indicating the interest level in the different resources. However, since the valuations are discontinuous, CAEI outcomes are not necessarily efficient, and so
our aim will be to compute CAEI allocations with improved welfare guarantees.  

\vspace{2mm}
The \emph{social welfare} of an
allocation $\vec{x}$ is the sum of utilities of all the agents: $SW(\vec{x}) = \sum_{i=1}^{n} V_i(\vec{x}_i)$.

\section{Multiple Divisible Goods} \label{sec:multiple_divisible}

We now formalize the model with multiple divisible goods and single minded agents. There is a set $N = \{1,
\ldots, n\}$ of agents and a set $M = \{1, \ldots, m\}$ of goods. Without loss of generality, each good comes in one unit that is
infinitely divisible. Each agent $i$ demands a bundle $D_i = \langle v_{i,1}, \ldots, v_{i,m} \rangle$, where $v_{i,j} \in [0,1]$
denotes the fraction required by agent $i$ from good $j$. The utility of agent $i$ for a bundle $\vec{z} \in [0,1]^{m}$ is
$V_i(\vec{z}) = 1$ if $z_j \geq v_{i,j}$ for all $ j \in M$, and $V_i(\vec{z}) = 0$ otherwise. Since every part of a good is equally
valued by agents, it suffices to specify its per unit price, {\em i.e.,} $p_j\in \mathbb R_+$ for good $j$.

We first show that the CAEI solution is guaranteed to exist for multiple divisible goods with single minded valuations, and can
moreover be computed in polynomial time. In order to prove this, we must introduce the \emph{Fisher market} model, which is a classical model of an economy developed by Fisher~\cite{Brainard00}.

A \emph{Fisher market} $\mathcal{M}$ consists of a set $N = \{1, \ldots, n\}$ of buyers (agents) and a set
$M = \{1, \ldots, m\}$ of divisible goods (items). 
Every buyer $i$ has:
\begin{itemize}
\item an initial budget $B_i > 0$, which can be 
viewed as some currency that can be used to acquire goods but has no intrinsic value to the buyer, and
\item a utility function $u_i: [0, 1]^m \rightarrow \mathbf{R}$ that maps a quantity vector of the $m$ items to a real value. $u_i(\vec{x}_i)$ represents the buyer's utility when receiving $\vec{x}_i$ amount of the items.  
\end{itemize}

Without loss of generality, the supply of each good is assumed to be one unit. A standard type of valuations is known as \emph{Leontief}, where the utility of a buyer $i$ for a bundle $\vec{x}_i$ is: $u_i({\vec{x}_i})=\min_{j \in [m]} \left\{ \frac{x_{ij}}{v_{ij}} \right\}$, where $v_{i,j}$ is the coefficient that describes buyer $i$'s valuation for good $j$. Thus buyers with Leontief utilities desire the goods in the same ratios (e.g. in the case of two goods---CPU and RAM---a buyer with coefficients $5$ and $1$, respectively, will require $5$ more units of CPU for every additional unit of RAM).

Each buyer in the market wants to spend its entire budget to acquire a bundle of items that maximizes its utility.  A market outcome is defined as a tuple $\langle \vec{p}, \vec{x}\rangle$, where
$\vec{p}$ is a vector of prices for the $m$ items and 
$\vec{x} = (\vec{x}_1, \ldots, \vec{x}_n)$ is
an allocation of the $m$ items, with $p_j$ denoting the price of item $j$ and $x_{ij}$
representing the amount of item $j$ received by buyer $i$. A market outcome that 
maximizes the utility of each buyer subject to its budget constraint
and clears the market is called a \emph{market equilibrium}~\cite{NRTV07}.
Formally, $\langle \vec{p}, \vec{x} \rangle$ is a market equilibrium if and only if:

\begin{enumerate}
\item Optimal bundle: $\forall i \in N$ and $\forall \vec{y}: \vec{y} \cdot \vec{p} \le B_i$, $\; u_i(\vec{x}_i)\geq u_i(\vec{y})$
\item Market clearing: Each good is fully sold or has price zero, i.e., $\forall j \in M$, $\sum_{j=1}^m x_{i,j}\le 1$, and equality
holds if $p_j>0$. Each buyer exhausts all its budget, i.e., $\forall i \in N, \ \sum_{j=1}^m x_{i,j}p_j = B_i$.
\end{enumerate}

A market equilibrium is guaranteed to exist for very general valuations~\cite{Maxfield97}. For Leontief utilities, it can be computed using
the Eisenberg-Gale (EG) convex program formulations that follows.

\begin{equation}\label{eq:ME}
\begin{array}{rcl}
\mbox{max } \ \ & & \sum^{n}_{i=1} B_i \cdot \log u_i \\
\mbox{s.t. } \ \ & & u_i \le \frac{x_{i,j}}{v_{i,j}},\ \forall i \in N, j \in M \\
  & & \sum^n_{i=1} x_{i,j}\leq 1,\ \ \forall \ j \in M \\
           & & x_{i,j}\geq 0,\ \ \forall \ i \in N, j \in M
\end{array}
\end{equation}

We can now show the existence and computation of the CAEI solution for single minded agents.

\begin{theorem} \label{thm:existence_divisible}
Given a fair division problem with multiple divisible goods and single minded agents, a CAEI solution is guaranteed to exist and can
be computed in polynomial time.  
\end{theorem}
\begin{proof} 
Given a fair division problem $(N, M, \mathcal{D})$, we consider a Fisher market where the set of buyers is $N$ and the set of items is $M$.
Each buyer $i$ has budget $B_i = 1$ and \emph{Leontief utility} described by the vector $D_i = \langle v_{i,1}, \ldots, v_{i,m}
\rangle$, i.e. for any bundle $z$, we have $$u_i(\vec{z}) = \min_{j \in M: v_{i,j} >
0} \left \{ \frac{z_j}{v_{i,j}} \right\}.$$ Recall that Fisher markets with Leontief utilities always have exact market equilibria.

Let $(\vec{x}, \vec{p})$ be any such equilibrium, where $x_{i,j}$ is the fraction received by agent $i$ from good $j$ and $p_j$ is the
price of good $j$. Denote the Leontief utility of buyer $i$ in the market by $u_{i}^{\mathcal{L}}(\vec{x}_i, \vec{p}) = \min_{j \in M:
v_{i,j} > 0} \left\{\frac{x_{i,j}}{v_{i,j}}\right\}$.  We argue the market equilibrium $(\vec{x}, \vec{p})$ is a CAEI solution for the
fair division problem with single minded agents. To this end, we must show that all the items are allocated, each agent spends no
more than a unit, and gets in exchange an optimal bundle at those prices. Clearly the first two requirements are met since $(\vec{x},
\vec{p})$ is a market equilibrium in the Fisher market with identical budgets.

We additionally show that each agent gets an optimal bundle in the fair division problem. For every agent $i$, if the allocation
$\vec{x}$ satisfies the property that $x_{i,j} \geq v_{i,j}$ for all $j \in M$, then $V_i(\vec{x}_i) = 1$ and $i$ gets its demand at
these prices. Otherwise, there is an item $k$ with $v_{i,k} > 0$ but $x_{i,k} < v_{i,k}$; then agent $i$ does not get its demand,
thus $V_i(\vec{x}_i) = 0$. Since $(\vec{x}, \vec{p})$ is a market equilibrium with respect to the Leontief utilities given by
$\vec{v}$, we have that $$u_{i}^{\mathcal{L}}(\vec{x}_i, \vec{p}) \leq \frac{x_{i,k}}{v_{i,k}} < 1.$$

Assume by contradiction that in the fair division problem agent $i$ could afford its demand set at prices $\vec{p}$, i.e. 
$\exists \vec{y} \in [0,1]^{m}$ with $y_{j} \geq v_{i, j}$ $\forall j \in M$, and $\vec{p}(\vec{y}) \leq 1$. Then in the Fisher
market with Leontief utilities, buyer $i$ could also purchase bundle $\vec{y}$ and get:
$$u_{i}^{\mathcal{L}}(\vec{y}, \vec{p}) =  \min_{j : v_{i,j} > 0} \left\{\frac{y_{j}}{v_{i,j}} \right\} \geq \min_{j : v_{i,j} > 0}
\left\{\frac{v_{i,j}}{v_{i,j}} \right\} = 1
> u_{i}^{\mathcal{L}}(\vec{x}_i, \vec{p}),$$
which is a strict improvement
over $\vec{x}_i$, contradicting that $(\vec{x}, \vec{p})$ is a market equilibrium in the Fisher market.
Thus the assumption was false and $(\vec{x}, \vec{p})$ is a CAEI.

Finally, a Fisher market equilibrium can be computed in polynomial time via a convex program formulation 
(Equation \ref{eq:ME}; see Codenotti and Varadarajan \cite{CV04} for more details). As argued above, this algorithm can be used to compute a CAEI solution for single minded agents. 
\end{proof}

\newpage
As the next example illustrates, the solution computed in Theorem~\ref{thm:existence_divisible} is not necessarily optimal.

\begin{example}
Consider 2 agents and 2 goods, where the demand of agent 1 is $D_1 = \langle 0.5, 0.4 \rangle$ and 
of agent 2 is $D_2 = \langle 0, 0.6 \rangle$.
The CAEI solution from Theorem~\ref{thm:existence_divisible} prices good 1 at $p_1 = 0$ and good 2 at $p_2 = 2$, and the allocation splits  good 2 equally between the two agents. This way agent 2 gets zero utility, while agent 1 gets utility 1 (and some extra good that is of no use).

Instead, another possible CAEI solution is to set prices $p' = \langle 1/3, 5/3 \rangle$, where agent 2 can now afford a quantity of
0.6 from good 2, and agent 1 still gets its demand. Thus the CAEI solution computed via the Leontief market equilibrium is dominated
by another CAEI solution with better welfare, which is not a solution to the program for Leontief market equilibrium.
\end{example}

We show that while maximizing social welfare is in general NP-hard, the problem can be solved in polynomial time when the number of
\emph{types} of agents is constant. This allows handling possibly large numbers of agents when their demands fit some standard
templates for resource requests.

\begin{theorem} \label{thm:SW_multiple_divisible_nphard}
Computing a CEEI solution that maximizes social welfare for single minded valuations with divisible goods is NP-hard.
\end{theorem}
\begin{proof}
We use a reduction from the NP-complete problem $\textsc{SET}$ $\textsc{PACKING}$: 
\begin{quote}
\emph{Given a collection $\mathcal{C} = \langle C_1, \ldots, C_n\rangle$ of finite sets and a positive integer $K \leq n$, does $\mathcal{C}$ 
contain at least $K$ mutually disjoint sets?
}
\end{quote}
Given collection $\mathcal{C}$ and integer $K$, let $\mathcal{M}$ be a fair division problem with agents $N = \{1, \ldots, n\}$, items $M = \{1, \ldots, m+n\}$, and demand sets such that each agent $i \in N$ wants $100\%$ of each item in the set $C_i \cup \{m+i\}$. In other words, the fair division problem on such instances can be seen as equivalent to that of allocating multiple indivisible goods (where $Q_j = 1$ for each good $j$ and agent $i$ can get either $0\%$ or $100\%$ of each good.

It can be checked that $\mathcal{M}$ has a CAEI solution with social welfare at least $K$ if and only if $\mathcal{C}$ has a disjoint collection of at least $K$ sets.
\end{proof}

\begin{theorem}
Given a fair division problem with single minded agents and divisible goods, a welfare maximizing CAEI solution can be found in polynomial time for a constant number of agent \emph{types}.
\end{theorem}
\begin{proof}
We are given a set $N = \{1, \ldots, n\}$ of agents, $M$ of items, valuations $\mathcal{D}$, and $\mathcal{T}$ a set of types of
agents, such that for each $i \in N$, $\tau(i) \in \mathcal{T}$ represents the type of agent $i$, and agents with same type have
identical valuation ($\tau(i)=\tau(i') \Rightarrow \vec{v_i} = \vec{v_{i'}}$). 
Algorithm \ref{alg:maxSW} solves this problem and runs in polynomial time for fixed $|\mathcal{T}|$. 

\begin{algorithm}[h!] \label{alg:maxSW}
  \SetAlgoLined\DontPrintSemicolon
\caption{\textsc{Compute-Max-CAEI}}
\KwData{Agents $N = \{1, \ldots, n\}$, items $M = \{1,$ $\ldots,$ $m\}$, valuations $\vec{v}$, types $\tau$}
\KwResult{Social welfare maximizing CAEI}
 $(\textsc{Opt}, \vec{x}^*, \vec{p}^*) \leftarrow (-\infty, \textsc{Null}, \textsc{Null})$ \\
\ForEach{$S \subseteq \mathcal{T}$} {
$(\vec{m}, \vec{p}) \leftarrow \textsc{Solve-Optimal-Subset}(S, \mathcal{M})$\\
$x_{i,j}=\frac{m_{i,j}}{p_j},\ \forall i\in N,\ \forall j \in M$\\
\If{$(\textsc{SW}(\vec{x}) > \textsc{Opt})$}{
$\textsc{OPT} \leftarrow \textsc{SW}(\vec{x})$\\
$(\vec{x}^*, \vec{p}^*) \leftarrow (\vec{x}, \vec{p})$\\
}
}
\textbf{Return} $(\vec{x}^*, \vec{p}^*)$
\end{algorithm}


\begin{center}
\begin{alignat*}{3}
    & \text{maximize  }   & \epsilon\  \\
    & \text{subject to} & \sum_{j=1}^{m} p_j \cdot v_{i,j} \leq 1, &  & \; \forall i \in N: \tau(i) \in S\\
           &           & \sum_{j=1}^{m} p_j \cdot v_{i,j} \geq 1 + \epsilon, & &\;  \forall i \in N: \tau(i) \in \bar{S}\\
             &         & m_{i,j} \geq p_j \cdot v_{i,j}, & & \;\ \  \forall j \in M, \forall i: \tau(i) \in S\\
		&      & \sum_{i=1}^{m} m_{i,j} = p_j, & & \;  \forall j \in M\\
		  &    & \sum_{j=1}^{m} m_{i,j} \leq 1, & & \; \forall i \in N\\
		  &    & \vec{p}\ge 0;\ \ \ \vec{m}\ge 0;\ \ \  
		  \epsilon \geq 0
  \end{alignat*}
\end{center}

Procedure \textsc{Solve-Optimal-Subset} is given as a linear program, where $p_j$ is the
price of good $j$, $m_{i,j}$ is the amount of money spent by agent $i$ on good $j$, and $\epsilon$ is a variable that should be
strictly positive in the optimal solution if the set $S$ of agent types can be made simultaneously happy. 

The observation underlying the algorithm is that agents of the same type must have identical utilities (but not necessarily identical bundles). The \textsc{Compute-Max-Caei} procedure tries to compute a CAEI solution for every set $S$ of types so that agents of type $\tau \in S$ are satisfied,
while agents outside $S$ cannot afford their demand. The LP constraints ensure these conditions are met and a positive solution is found if and only if there exists a CAEI solution for the set $S$. At the end, the \textsc{Compute-Max-Caei} algorithm takes the maximum over all $S$. 
\end{proof}

\section{Cake Cutting} \label{sec:cake_cutting}

Next we investigate the cake cutting problem with single minded agents. 
Cake cutting models the problem of allocating a heterogeneous divisible resource, such as land, time, mineral deposits or computer memory, among agents with different preferences. The ``cake'' is represented mathematically as the interval $[0,1]$ and the agents have different preferences over the interval. A ``piece'' of cake $A$ is a union of intervals: $A = (I_1, \ldots, I_m)$.

The literature on cake cutting has seen very interesting algorithmic developments in recent years (see, e.g.,~\cite{KLP13,SHA15,AM16}), all of which are concerned with additive valuations in the one-dimensional model. Among the 
few exceptions we mention: \cite{CLP11} studied fair division with additive valuations constrained by a minimum length requirement (PUML),
\cite{BPZ13} studied externalities in cake cutting, while \cite{HHA15,HNHA15},  studied cake cutting in two dimensions, where the agents also care about the shape of the pieces that they receive. In earlier work, a different multidimensional model of fair division has been explored under the name of pie-cutting by Brams et al.~\cite{BJK08}, while a very general model of fair resource division has been explored in the works of Dall’Aglio and Maccheroni~\cite{DM09} and of Husseinov and Sagara~\cite{HS13}, who prove the existence of fair and efficient allocations under mild conditions assuming that the valuations are continuous. However, single minded valuations are discontinuous, and so their results do not directly imply ours.

In this paper we focus on single minded agents, 
and so each agent $i$ will demand a set $D_i$ consisting of possibly several disjoint intervals: $D_i = (D_{i,1}, \ldots, D_{i,
m_i})$. The utility of agent $i$ for a piece of cake $A \subseteq [0,1]$ is: $V_i(A) = 1$ if $D_i \subseteq A$ and $V_i(A) = 0$
otherwise.

Examples of single minded agents in cake cutting include cases where the agent requires land that has buildings on it (with some particular
functionality) or wants to build a house in a particular location. The case where each single minded agent requires a contiguous piece ({\em i.e.,} $m_i=1$)
can in fact be mapped to a standard problem known as ``interval scheduling'' in operations research.
An instance of 3 agents with single minded valuations in cake cutting is illustrated in Figure~\ref{pic:cake}. 

\begin{figure}[!t]
    \centering
    \includegraphics[width=0.43	\textwidth]{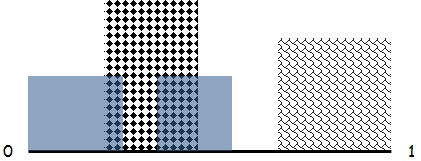}
    \caption{Three agents with single minded valuations; the intervals required by each agent are shown through shaded areas of
    different filling.} \label{pic:cake}
\end{figure}

\vspace{2mm}
We show that a CAEI solution always exists and can be computed efficiently in this model.

\begin{theorem}
A CAEI solution is guaranteed to exist and can be computed in polynomial time in the cake cutting problem with single minded agents.
\end{theorem}
\begin{proof}
For each $i \in N$, let $D_i = (D_{i,1}, \ldots, D_{i, m_i})$ be the disjoint intervals in its demand set, where $D_{i,j} = [l_{i,j},
r_{i,j}] \subseteq [0,1]$. Divide each interval $D_{i,j}$ into two equal pieces, and let $F_{i,j} = \left\{l_{i,j}, \frac{l_{i,j} +
r_{i,j}}{2}, r_{i,j}\right\}$ be the set of points resulting from this partition of $D_{i,j}$. Consider the set of all the distinct
points thus obtained: 
$\mathcal{P} = \left( \bigcup_{i=1}^{n}\right.$$\left. \bigcup_{j=1}^{m_i}\right.$$\left.F_{i,j} \right) \cup \{0, 1\}$ $=$ $\{q_0$,
$q_1$, $\ldots$, $q_m\}$, where $0=q_0$ $<$ $q_1$ $<$ $\ldots$ $<$ $q_m = 1$. Then we can view every segment $I_k = [q_k, q_{k+1}]$
delimitated by two consecutive points in
$\mathcal{P}$ as an indivisible good; moreover, good $I_k$ is in the demand set of each agent $i$ for which $[q_k, q_{k+1}] \subset
D_i$. This gives a related allocation problem with $m$ indivisible goods and $n$ agents, where no agent has a demand set consisting
of a single good. The latter property follows from the fact that we partitioned each interval $D_{i,j}$ into two pieces, which are
viewed as distinct goods. From~\cite{BHM15}, a competitive equilibrium from equal incomes is guaranteed to exist and can be computed in
polynomial time for every fair division problem with indivisible goods and single minded agents, where no agent has a demand set with
only one good; let $(\vec{x}, \vec{p})$ be such an allocation and prices.

Then we can compute a CAEI solution in cake cutting by allocating the intervals in $\mathcal{P}$ in the same order (from left to right)
as the indivisible goods are allocated under $\vec{x}$ and setting the price curve of each interval $[q_k, q_{k+1}]$ uniformly such
that it sums up to the price of the corresponding indivisible good.  
\end{proof}

\vspace{2mm}
When the demands of the agents are contiguous, a welfare maximizing CAEI solution can be computed in polynomial time; the proof exploits a connection with the interval scheduling problem.

\begin{theorem} \label{thm:cake_contiguous_SW}
A welfare-maximizing CAEI can be computed in polynomial time in cake cutting with single minded agents and contiguous
demands. Moreover,  when there
are no identical agents, the welfare maximizing CAEI coincides with the solution to the pure optimization problem of maximizing welfare.
\end{theorem}
\begin{proof}
The idea is to leverage a connection with the interval scheduling problem~\cite{KT05} to first decide the agents that get allocated,
and then to construct an \emph{asymmetric} pricing scheme to implement a price equilibrium.

The interval scheduling problem is as follows. There are $n$ jobs to be run on a supercomputer, where each job $i$ runs from time $s_i$ to time $f_i$. There are multiple such requests arriving
simultaneously, and the goal is to process as many of them as possible, but the computer can only run one job at a time. The question
is how to schedule the jobs so that the maximum number of requests is served. The optimal allocation---that maximizes the number of jobs scheduled---is given by a greedy algorithm: Schedule first the job
with the earliest finishing time, then remove the jobs intersecting with it, and repeat among the remaining jobs.

This problem can be mapped to cake cutting with single minded agents by considering an agent for each job,
i.e. $N = \{1, \ldots, n\}$, and setting the demand of each agent $i$ to the schedule of job $i$, that is $D_i = \{[s_i, f_i]\}$,
where $s_i$, $f_i$ are normalized in $[0,1]$.

While our aim is to indeed maximize the number of completed jobs, we want to achieve this in a way that is fair. That is, we
will compute an allocation of the cake $\vec{x}$ and a price curve $\vec{p}$ such that $(\vec{x}, \vec{p})$ represent a CAEI solution,
where $p(x)$ is the value of the price density function at point $x$. 
This can be accomplished by allocating the agents in the same name order that the greedy algorithm allocates the corresponding jobs, with the caveat
that if multiple agents have the same finishing time, then the one with the \emph{latest} starting time must be selected.
The agents that intersect with the last scheduled agent will not get their demand set.

For each $k$-th agent that gets selected to receive their full demand, we select two very small intervals at the left and right endpoints of their demanded interval, respectively, and price 
them (uniformly) to sum up to $k \cdot \epsilon$ and $1 - k \cdot \epsilon$, respectively.
When there are no identical demands, the agents that don't get served but for which there is an unallocated piece of their demand, can receive a 
small such piece (at a total price of 1); otherwise, they get nothing (but their demand now has the property that it costs more than 1). This allocation 
can be done so that the whole cake is divided.
 
Then to handle identical agents, the algorithm can be modified such that whenever selecting the earliest finishing time and finding
multiple agents with the same finish \emph{and} start time, then take a very small interval at the leftmost point of the unallocated
cake and divide it equally among the identified identical agents, setting the price uniformly to sum up to their budgets. Then iterate
on the remaining cake. 
\end{proof}

An example with four agents is given in Figure 2; note that in this picture, for simplicity of presentation, the $k$th scheduled agent has the areas costing $k \epsilon$ and $1 - k \epsilon$ distributed over its entire demand---this is allowed by the algorithm. However, generally one can select two arbitrarily small intervals (prefix and suffix of agent $k$'s demand) and only price these.

\begin{figure}[!t]
    \centering
    \includegraphics[width=0.8\textwidth]{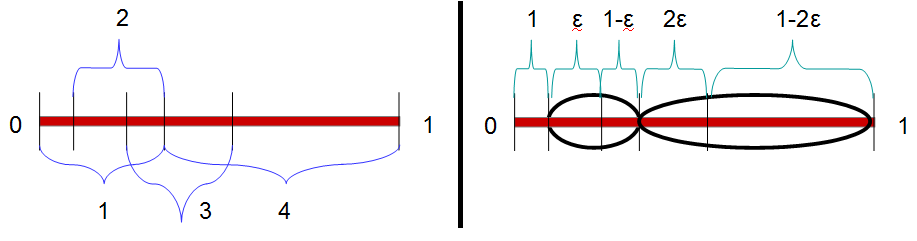}
    \caption{Four agents with contiguous single minded demands, indicated by the blue brackets on the intervals (left). On the right, the two agents selected in the optimal schedule are circled and the prices marked on top of each differently priced interval (the price listed indicates the price of the whole interval for the corresponding green bracket).} \label{pic:cake}
\end{figure}

\vspace{2mm}
On the other hand, if the demands are not contiguous, then the problem of computing a welfare maximizing CAEI is NP-hard. 
\begin{theorem} \label{thm:SW_cake_nphard}
Given a cake cutting problem with single minded valuations, computing a welfare-maximizing CAEI is NP-hard.
\end{theorem}
\begin{proof}
We analyze the decision version of the welfare maximization problem and reduce from the NP-complete problem \textsc{Set Packing}: \emph{Given universe $\mathcal{U} = \{1, \ldots, m\}$, collection $\mathcal{C} = \langle C_1, \ldots, C_n \rangle$ of finite subsets of $\mathcal{U}$, and a positive integer $K \leq n$, does $\mathcal{C}$ contain at least $K$ mutually disjoint sets?}

Given such an instance, consider a cake cutting problem with single minded valuations where the set of agents is $N = \{1, \ldots, n\}$ and the cake is divided into $m + n$ intervals, such that interval $I_j = \left[\frac{j}{m+n}, \frac{j+1}{m+n}\right]$ and the demand set of each agent $i$ is 
$D_i = \left(\bigcup_{j \in C_i} I_j \right) \cup I_{m+i}$. The idea is that the intervals $I_1, \ldots, I_m$ correspond to the elements of the universe $\mathcal{U}$ and the demand of agent $i$ corresponds to the elements of $C_i$ plus one more interval ($I_{m+i}$) specifically designed for agent $i$. 

If the cake cutting problem has an optimal allocation from equal incomes with social welfare of at least $K$, then clearly it is possible for $K$ agents to simultaneously receive their demand sets, and so the \textsc{Set Packing} problem has a solution of value at least $K$. On the other hand, if the \textsc{Set Packing} problem has a solution of value at least $K$, let $\mathcal{S} = \langle C_{i_1}, \ldots, C_{i_K} \rangle$ be the respective sets in the solution. Then consider the following allocation and price of the cake. For each agent $i$ for which $C_i \subset \mathcal{S}$, price all the intervals $I_j$ where $j \in C_i$ uniformly and at a total value of $1$; moreover, set the price of interval $I_{m+i}$ at zero. For each agent $i$ for which $C_i \not \subseteq \mathcal{S}$, price the interval $I_{m+i}$ uniformly at total value $1$ and give it to agent $i$. Price all the remaining cake (that has not been allocated by this point) at zero and allocate it arbitrarily. It can be checked that the constructed allocation and prices represent an optimal allocation from equal incomes. Thus the \textsc{Set Packing} problem has a solution of value at least $K$ if and only if the cake cutting problem has an optimal allocation from equal incomes with social welfare at least $K$, which completes the reduction.
\end{proof}

We can also maximize social welfare for discontinuous demands in polynomial time when the number of agents is constant. 
The idea is to try out every possible set $S \subseteq N$ of \emph{happy} agents, who will receive their required land at a price of one,
while the remaining agents split the remaining cake at prices so they cannot afford their demand. Then select the outcome that maximizes welfare 
among all choices of $S$. We note that the procedure of computing a partition induced by the desired intervals was also employed by Cohler et al.~\cite{CLPP11} in the context of computing welfare maximizing envy-free allocations for additive piecewise constant valuations.

\begin{theorem}
Given a cake cutting problem with single minded agents, a welfare maximizing CAEI can be computed in polynomial time when the number of agents is constant.
\end{theorem}
\begin{proof}
We first show that given a fixed set of agents $S \subseteq N$, it can be checked in polynomial time if there exists a CAEI solution that makes all the agents in $S$ happy, and all the agents in $N \setminus S$ unhappy. Consider the partition $\mathcal{P}$ induced by the endpoints of the consecutive intervals desired by all the agents; let $\mathcal{P} = (I_1, \ldots, I_m)$. Divide each interval $I_j$ into a number of pieces (of the same length) equal to the number of agents that want this piece. This gives a refinement $\mathcal{P}' = (J_1, \ldots, J_r)$ of the partition $\mathcal{P}$. Then we can view each interval $J_{\ell}$ as a divisible good and run the linear program used as a subroutine in Algorithm \ref{alg:maxSW} on these goods with the set of agents $N$, where the demand of $i$ for good $J_r$ is set to $v_{i,r} = 1$ if $J_r \subseteq D_i$ and $v_{i,r} = 0$, otherwise. Then we obtain a feasible solution if and only if the cake cutting problem has a CAEI allocation where the set $S$ is happy. Note that the increase in the number of goods is polynomial in $n$ and the sizes of the original desired intervals, $D_i$.

By iterating over every possible set $S$ of agents, we can select the one that maximizes social welfare in polynomial time when the number of agents is constant.
\end{proof}

Finally, we illustrate the fact that in the cake cutting domain with single minded valuations, the CEEI solution---which requires that each agent completely exhausts its fictitious budget---becomes very different from the CAEI notion; in fact it achieves worse social welfare. Consider the following problem.

\begin{example}
Let there be a cake cutting instance with agents $N = \{1, 2, 3\}$ and valuations such that the desired intervals are $D_1 = [0, 0.5]$, $D_2 = [0.5, 1]$, 
and $D_3 = [0,1]$. Then a CAEI solution is to set the price uniformly, such that $p(x) = 2$ for all $x \in [0,1]$, and the allocation to $\vec{x}_1 = [0,0.5]$, $\vec{x}_2 = [0.5, 1]$, and $\vec{x}_3 = \emptyset$. Then agents $1$ and $2$ get precisely their demand, at a price of $1$, while agent $3$ gets the empty set; however this is feasible since agent $3$ cannot do better---the price of its demand, $D_3$, exceeds the agent's budget: $p(D_3) = 2 > 1$. Moreover, this coincides with the optimal schedule from the interval scheduling problem, where the start and end times are $s_1 = 0$, $f_1 = 0.5$, $s_2 = 0.5$, $f_2 = 1$, and $s_3 = 0$, $f_3 = 1$.

Clearly, this is not a CEEI solution because agent $3$ does not exhaust its budget. Instead, a CEEI solution would be to set $p(x) = 2$, $\forall x \in [0, 0.5]$, $p(x) = 4$, $\forall x \in [0.5, 1]$, and the allocation to $\vec{x}_1 = [0, 0.5]$, $\vec{x}_2 = [0.5, 0.75]$, and $\vec{x}_3 = [0.75, 1]$. This outcome is worse than the optimal schedule, supported at the welfare maximizing CAEI solution computed above.
\end{example}


\section{Multiple Discrete Goods} \label{sec:multiple_discrete}
Our setting for discrete goods is as follows. There is a set $N = \{1, \ldots, n\}$ of agents and $M = \{1, \ldots, n\}$ of items,
where each item $j$ comes in $Q_j$ indivisible copies. Each agent $i$ has a demand set $D_i \subseteq M$, and wants a copy of each
item $j \in D_i$; $D_i$ is not a multi-set. Without loss of generality, we assume that each type of item is required by some agent. We
illustrate the model with an example.

\begin{example}
There are agents $N = \{1, 2\}$, items $M = \{1, 2, 3\}$, quantities $Q_1 = Q_3 = 1$, $Q_2 =
2$, and demands $D_1 = \{1, 2\}$, $D_2 = \{2, 3\}$. 

The goal is to find a set of prices for each item---such that every copy from the same item has an identical price---together with an
allocation $\vec{x}$ such that $(\vec{x}, \vec{p})$ is a CAEI solution. In our example, a CAEI solution is attained at: $p_1 = 0.5, p_2
= 0.5, p_3 = 0.5$ and $\vec{x}_1 =\{1,2\}$, $\vec{x}_2 = \{2,3\}$.
\end{example}

Note that similarly to the cake cutting problem, the CAEI solution is very different from CEEI. In particular, there are instances---such as the one illustrated in the next example---where although there are enough items to go around to all the agents, the number of units is ``wrong'', and so the prices can never be set so that all the budgets are exhausted. This provides additional motivation for studying the CAEI notion of fairness when allocating discrete goods.

\begin{example}
Consider an instance with agents $N = \{1,2,3,4\}$,
itms $M = \{1,2\}$, quantities $Q_1 = Q_2 = 3$, demand sets
$D_1 = \{2\}$, $D_2 = \{1,2\}$, $D_3 = \{1,2\}$, $D_4 = \{1\}$.
Then the prices $p_1 = p_2 = 0.5$ and allocation $\vec{x}_1 = \{2\}$, $\vec{x}_2 = \vec{x}_3 = \{1,2\}$, $\vec{x}_4 = \{1\}$ represents a CAEI solution.

However, there is no CEEI solution, since it's never possible to allocate the entire set of items in a way that exhausts all the budgets.
\end{example}

Next we provide a succinct characterization of the instances that admit a CAEI solution. To this end, a useful notion will be that of
\emph{over-demand}. We say that an item $j$ is over-demanded among a set of agents $S$ if the aggregate demand of the agents in $S$
exceeds the available supply from item $j$, namely $Q_j$. We illustrate this phenomenon in the next example.

\begin{example} 
Consider a market $\mathcal{M}$ with single minded agents and discrete goods, where $M = \{1,2\}$, with quantities $Q_1 = 2$ and $Q_2=
4$, and agents $N = \{1, \ldots, 4\}$, with demands: $D_1 = D_2 = D_3 = \{1\}$, $D_4 = \{1,2\}$.  Then the aggregate demand of the set
of agents $S = \{1,2,3\}$ with singleton demands, consists of 3 copies of item 1, which exceeds $Q_1$, the available supply from this
item. Thus item $1$ is over-demanded among $S$.
\end{example}

The next theorem shows that these are essentially the only instances that don't have a fair outcome. We say that demand of agent $i$
is singleton if $|D_i|=1$. 

\begin{figure}
    \centering
    \includegraphics[width=1.05 \textwidth]{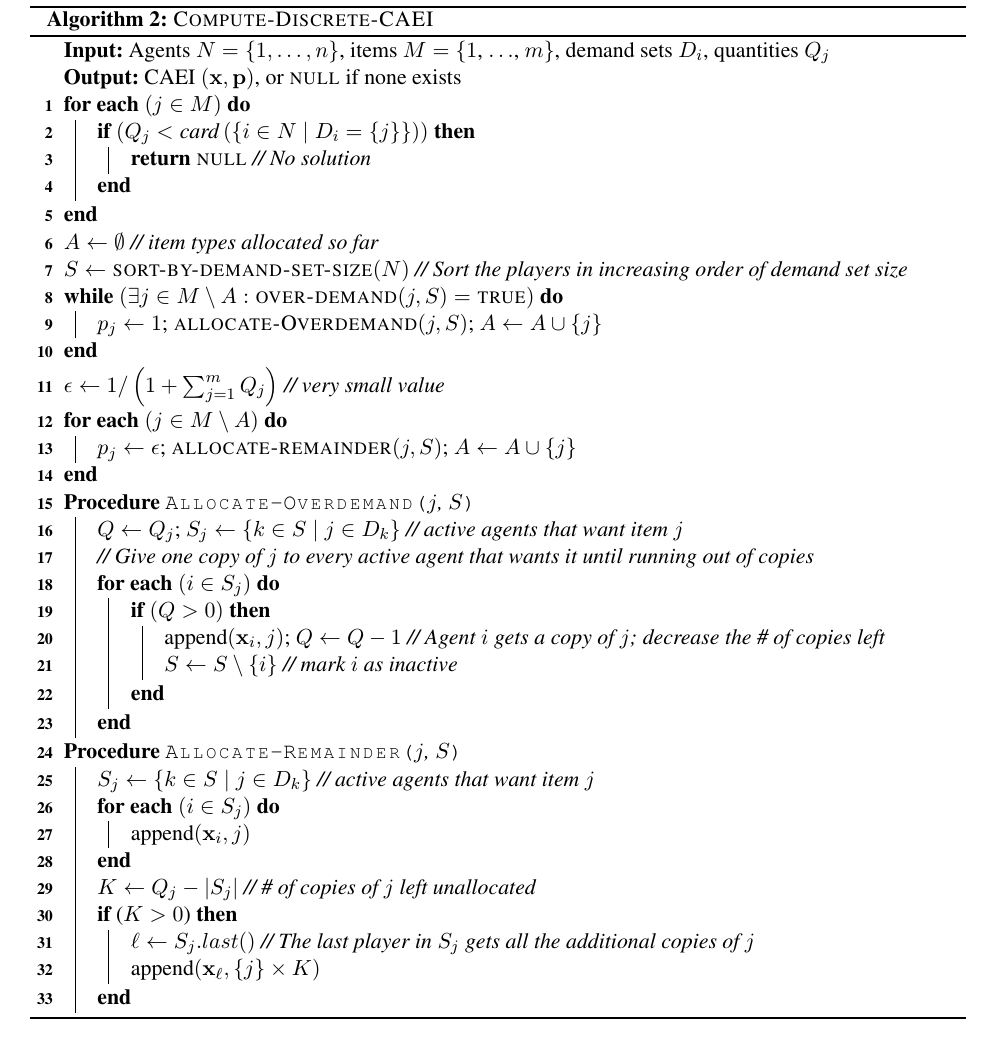}
   \label{alg2}
   \caption{Multiple discrete goods and single minded valuations: Polynomial-time algorithm yielding a succinct characterization of instances that admit a CAEI solution}
\end{figure}

\begin{theorem} \label{thm:discrete_characterization}
Given a fair division problem with single minded agents and discrete goods, a CAEI solution exists if and only if 
there is no set of agents with singleton demands that give rise to an over-demanded item among that set. 
Moreover, a solution can be computed in polynomial time if it exists.
\end{theorem}
\begin{proof}
Consider Algorithm 2. Clearly, if there is a set of agents $S$ with singleton demands, such that $D_i = \{j\}$ for some $j \in M$ and all $i \in S$, where
$|S| > Q_j$, then no CAEI can exist. This is because $p_j\le 1$ causes over-demand, and $p_j>1$ makes it un-affordable and
thereby un-sold. 

We claim all other instances admit a CAEI. Algorithm 2 forms a sequence $S$ of active agents, initially containing all the agents sorted in increasing order of $|D_i|$, and then iterates over the item types, at each step searching for an item that is desired by more agents among the active ones (i.e. in $S$) than there are copies available. If such an item $j$ is found, then $j$ is allocated at a price of $p_j = 1$, one copy to each of the first $Q_j$ agents in $S$ that want it. All the active agents that got item $j$ are removed from $S$, and the algorithm searches for another such over-demanded \emph{and} unallocated item among the updated set $S$. Otherwise, when no such item is found, all the goods left are in abundant supply for the agents in $S$ and can be given for a small price ($\epsilon$) among them. Note that these last agents, which are allocated bundles at a price of $\epsilon$ per good, receive not only items that they desire, but in fact all the remaining items.

We argue that the pricing rule ensures that a CAEI is computed. In particular, all the goods are allocated---at latest, the player (or set of players), which receive items in the last round of the algorithm, get all the remaining items. Moreover, the pricing rule is such that each agent $i$ is allocated all of its bundle at once, regardless of whether $\vec{x}_i$ contains $i$'s required demand set or not. The agents are divided in two categories, namely those who get exactly one item at a price of $1$ (Line 9) and those who get multiple items at a price of $\epsilon$ (Line 13). In both cases, the pricing rule ensures that the allocation $\vec{x}_i$ costs at most $1$.

Finally, we argue that each agent gets an optimal bundle at the given prices. Let $i$ be an agent 
that has \emph{not} received its demand under the allocation $\vec{x}$ ($D_i \not \subseteq \vec{x}_i$). 
The algorithm gives all the agents with singleton demands their required items, so we have
$|D_i| > 1$. We show that agent $i$ cannot afford the set $D_i$:

\begin{description}
\item[\emph{Case 1}).] Agent $i$ receives one item priced at $1$. Since all the items have positive prices and $|D_i| > 1$, we have
$\vec{p}(D_i) > 1$, so $i$ cannot afford its demand.
\item[\emph{Case 2}).] Agent $i$ gets an empty bundle. Then agent $i$ never got allocated in the $\textsc{Allocate-Remainder}$ subroutine, which corresponds to the phase where the remaining items are plentiful for the currently unallocated agents. Thus agent $i$ wanted at least one item $j$ allocated in one of the calls to the $\textsc{Allocate-Overdemand}$ subroutine, which means that $p_j = 1$. From $|D_i| > 1$, we have $\vec{p}(D_i) > 1$ and $i$ cannot afford its demand set.
\item[\emph{Case 3}).] Agent $i$ receives several items priced at $\epsilon$. If $D_i$ contains any item that
was allocated at a price of $1$, then clearly buyer $i$ cannot afford $D_i$. Otherwise, none of the items in $D_i$ was labeled as
over-demanded; then buyer $i$ both affords and receives at least one copy of each item in its demand set, which would contradict the
choice of agent $i$ as unhappy.  
\end{description}

From the case analysis we conclude that each agent $i$ gets an optimal bundle at the current prices, and so $(\vec{x}, \vec{p})$ is a CAEI solution. 

The algorithm has at most $m$ rounds, each of which calls the $\textsc{allocate}$ procedure depending on whether there is over-demand or under-demand. In the first type of call, the procedure executes a constant number of operations for each agent in the set of active agents that want the current item, which is clearly bounded by $c \cdot |S|$, where $c$ is a constant and $|S| \leq n$. In the second type of call, of under-demand, the procedure allocates all the players in $S$ one item each, and then gives the last player allocated in that call all the remaining items; this takes $O(n)$ operations.
Thus the runtime can be bounded by $poly(m,n,\log{Q_j})$.\\

It follows that in both cases, the algorithm runs in polynomial time, which completes the proof.
\end{proof}

\vspace{4mm}
We illustrate the execution of Algorithm 2 through an example.
\begin{example}
Let $N = \{1$,$2$,$3$,$4$,$5\}$, $M = \{1$,$2$,$3$,$4$,$5\}$, quantities $Q_1 = 2$, $Q_2 = 4$, $Q_3 = 2$, $Q_4 = 3$, $Q_5=2$, demands:
$D_1 = \{1\}$, $D_{2} = \{1$,$2\}$, $D_3 = \{1$,$3\}$, $D_4 = \{2$,$3$,$4\}$, $D_4 = \{2$,$3$,$4$,$5\}$. Only agent $1$ has singleton
demand $\{1\}$ and $|Q_1| \geq 1$, so there is a CAEI. 

Algorithm 2 sorts the agents in increasing order by demand set size and
initializes $S = (1,2,3,4,5)$. The search for over-demanded items begins. Item $1$ is wanted by 3 agents in $S$, namely $\{1,2,3\}$,
but $Q_1 = 2 < 3$; it is over-demanded. Set $p_1 = 1$, $\vec{x}_1 = \vec{x}_2 = \{1\}$. Update $S \leftarrow S \setminus \{1,2\} =
\{3,4,5\}$.  

Next item $3$ is over-demanded (wanted by agents $3,4,5$ but $Q_3 = 2$). Set $p_3 = 1$ and $\vec{x}_3= \vec{x}_4=\{3\}$.
Update $S = \{5\}$. 

All the items are in sufficient quantities now (i.e. there is no over-demand any more). Give all the remaining goods to agent $5$ at price $\epsilon=1/14$ each.
\end{example}


It is NP-hard to compute a welfare-maximizing CAEI (see, e.g., \cite{BHM15}). However, if we no longer insist that \emph{all} the items are allocated,
we can get a polynomial time algorithm by solving the problem as if the goods were divisible, and then rounding the solution.

\begin{theorem} \label{thm:max_SW_discrete}
Consider a fair division problem with agents $N$, $m$ items in quantities $Q= (Q_1, \ldots, Q_m)$, and demand 
$D_i \subseteq \{1,\ldots, m\}$ for agent $i\in N$. Let $\mathcal{T}$ be the types of the agents, where same type agents have the same
demand. 
Then a welfare maximizing CAEI, where all items need not be sold, can be
computed in polynomial time if $|\mathcal T|$ is a constant. 
\end{theorem}
\begin{proof}
We can convert the problem into an instance with divisible goods, where there is one unit from every good and the percentage required
by player $i$ from item $j$ is $\frac{D_{i,j}}{Q_j}$. Consider the variant of Algorithm 1 where it is not necessary that all the items
are sold---note this variant can be obtained with minor modifications, by changing the respective constraints from ``$=$'' to ``$\leq$''.
By running such an algorithm, we obtain an optimal
allocation from equal incomes with divisible goods without item clearing and which is obtained at some tuple $(\vec{x}, \vec{p})$. Then
define a corresponding allocation $\bar{\vec{x}}$ in the instance with discrete goods, where $\bar{\vec{x}}_{i,j} = \lfloor x_{i,j}
\cdot Q_j\rfloor$. Then the tuple $(\bar{\vec{x}}, \vec{p})$ satisfies the property that the allocation is feasible,
each player gets an optimal bundle at the current prices, and the social welfare is the same as in the corresponding problem with
divisible goods. Moreover, 
no allocation in the problem with discrete goods can obtain higher social welfare than in the corresponding problem with divisible
goods. 
It follows that we obtain a welfare maximizing allocation from equal incomes via Algorithm
1, which runs in time $poly(n, m)$.
\end{proof}

\section{Relation to Other Fairness Solution Concepts}

In this section we show that the CAEI solution is not equivalent to other standard concepts of fairness. Clearly CAEI is not equivalent to proportionality in the realm of single minded valuations, since proportionality exists very rarely on such instances. However, envy-free allocations often do exist in the domain that we studied. We wish to show that the CAEI solution is strictly stronger than envy-freeness; that is, while every CAEI allocation is envy-free, not every envy-free allocation can be supported through CAEI. 

Thus CAEI can be seen as a method for ruling out the least desirable envy-free allocations, i.e. those that cannot be ``explained'' through any equilibrium pricing scheme. It remains an interesting direction for future work to more fully understand the CAEI solution.
 
\begin{theorem} \label{thm:non_equivalence}
The CAEI solution is strictly stronger than envy-free freeness for single minded agents in all the models considered: multiple divisible goods, cake cutting, and multiple discrete goods.
\end{theorem}
\begin{proof}
The proof is broken down in three components, corresponding to the scenarios considered. For multiple divisible goods, let there be an instance with $N = \{1, 2\}$, $M = \{1, 2\}$, and demands $D_1 = \langle 0.2, 0.2 \rangle$, $D_2 = \langle 0.8, 0.8 \rangle$. Consider the allocation $\vec{x}_1 = \langle 1,0 \rangle$, $\vec{x}_2 = \langle 0, 1 \rangle$. Clearly, this allocation is envy-free since each agent is missing one good in either of the bundles $\vec{x}_1$, $\vec{x}_2$. Assume there exist supporting CAEI prices $\vec{p}$. By feasability constraints, we have: $p_1 \leq 1$, $p_2 \leq 1$ (I). From the optimality condition of CAEI, we must have that no agent can afford a better bundle; in particular, agent $1$ cannot afford their own demand set: $0.2 \cdot p_1 + 0.2 \cdot p_2 > 1$ (II). 
From (I) and (II), we have that $5 < p_1 + p_2 \leq 2$, contradiction. 

In the cake cutting model, let $N = \{1, 2\}$, $D_1 = [0, 0.4]$, $D_2 = [0.4, 1]$, and consider the envy-free allocation $\vec{x}_1 = [0, 0.2] \cup [0.4, 0.7]$, $\vec{x}_2 = [0.2, 0.4] \cup [0.7, 1]$. Suppose there exists a CAEI price curve $\vec{p}$ at $\vec{x}$; denote by $p_1, p_2, p_3, p_4$ the prices of the intervals $[0, 0.2]$, $[0.2, 0.4]$, $[0.4, 0.7]$, $[0.7, 1]$, respectively. By feasability constraints, we get: $p_1 + p_3 \leq 1$ and $p_2 + p_4 \leq 1$ (III). Both agents have a utility of zero at $\vec{x}$, so none should afford their demand set, i.e. $p_1 + p_2 > 1$ and $p_3 + p_4 > 1$ (IV). From (III) and (IV) we get that $2 < p_1 + p_2 + p_3 + p_4 \leq 2$, contradiction. Thus $\vec{p}$ cannot exist.

Finally, for multiple discrete goods, let $N = \{1, 2\}$, $M = \{1, 2, 3, 4\}$ with $Q_j = 1$ $\forall j$, demands $D_1 = \{1,2\}$, $D_2 = \{3,4\}$. Consider allocation $\vec{x}_1 = \{1, 3\}$, $\vec{x}_2 = \{2,4\}$. Both agents have utilities zero at either bundle $\vec{x}_1, \vec{x}_2$. The CAEI constraints give: $p_1 + p_3 \leq 1$, $p_2 + p_4 \leq 1$ (V) and $p_1 + p_2 > 1$, $p_3 + p_4 > 1$ (VI). It follows that $2 < p_1 + p_2 + p_3 + p_4 \leq 2$, contradiction.

In all the scenarios considered, the envy-free allocation $\vec{x}$ is not a CAEI solution.
\end{proof}

\section{Discussion}
We studied the computation and complexity of the {\em competitive allocation from equal incomes}, a method for fair
division, for single minded agents for multiple divisible and discrete goods, and cake-cutting. Although a solution can be computed
efficiently, social welfare maximizing solutions are hard to compute in general. However, we solved the latter efficiently for
several
interesting
special scenarios. Problems with constantly many goods, or altogether characterizing easy instances remain unresolved. 
It will be interesting to settle these. We note that our results also work valuation classes where the agents have a desired target set (e.g. the set of parts of a bike),
and very small positive value ($\epsilon$) for any other bundle of goods. For such scenarios, the solution computed is an $\epsilon$-CAEI.

An immediate generalization is that of Leontief valuations, where even existence of succinct characterization of 
instances that admit CEEI or CAEI in case of discrete goods is not clear, while the bigger question of 
understanding fair division with complementarities remains a mystery at large.

\section{Acknowledgements}

We are grateful to Erel Segal-Halevi for useful feedback.

\nocite{*}

\bibliographystyle{plain}


\appendix

\section{Missing Proofs}

\begin{theorem} \label{thm:divisible_given_xp}
Given a fair division instance with single minded valuations and divisible goods, and allocation $\vec{x}$, a CAEI solution at $\vec{x}$ can be found in polynomial time (if one exists). Similarly, given price vector $\vec{p}$ a matching allocation can be found in polynomial time.
\end{theorem}
\begin{proof} 
Given $\vec{p}$, let $S$ be the set of agents that can afford their demand at these prices; that is, for each $i \in S$, we have that  $\sum_{j=1}^{m} v_{i,j} * p_j \leq 1$. Then an equilibrium allocation can be found by solving the following constraints with variable $\vec{x}$:
\begin{small}
\begin{alignat*}{2}
		      & \sum_{j=1}^{m} x_{i,j} \cdot p_j \leq 1, & & \forall i \in N\\
                      & \sum_{i=1}^{n} x_{i,j} = 1, & & \forall j \in M\\
		      & x_{i,j} \geq v_{i,j}, & & \forall i \in S\\
		      & x_{i,j} \geq 0, & & \forall i \in N, \forall j \in M\\
\end{alignat*}
\end{small}
That is, we require that each agent can afford its own bundle, the agents in $S$ get their demand set, all the goods are allocated and the allocation is feasible.

For the second part of the theorem, suppose we are given an allocation $\vec{x}$. Let $S \subseteq N$ be the set of agents that receive their demand under $\vec{x}$, i.e. for which $x_{i,j} \geq v_{i,j}$ $\forall j \in M$. If $\vec{x}$ satisfies basic feasability constraints (i.e. all the items are sold and $\vec{x} \geq 0$), Then a price vector that matches the allocation (if it exists) can be found by solving the following linear program:

\begin{small}
\begin{alignat*}{2}
    \text{max  }   & \; \; \epsilon\  \\
    \text{subject to} & \; \; \sum_{j=1}^{m} p_j \cdot x_{i,j} \leq 1, &  & \; \forall i \in S\\
                      &  \; \; \sum_{j=1}^{m} p_j \cdot x_{i,j} \geq 1 + \epsilon, & &  \; \forall i \in \bar{S}\\
                      & \; \; p_j \geq 0, & &  \;  \forall j \in M\\
		      & \; \; \epsilon \geq 0, & & 
\end{alignat*}
\end{small}
The constraints require that the agents in $S$ can afford their bundles, while the agents in $\bar{S}$ cannot. The value of $\epsilon$ obtained by solving the LP is strictly positive if and only if there is an equilibrium at $\vec{x}$ (and zero otherwise).
\end{proof}

\begin{theorem} \label{thm:cake_given_xp}
Given a cake cutting problem with single minded valuations and allocation $\vec{x}$, a supporting price curve can be computed in polynomial time (if one exists). Similarly, given a price curve $\vec{p}$, a CAEI allocation at $\vec{p}$ can be computed in polynomial time.
\end{theorem}
\begin{proof}
Given a price curve $\vec{p}$, let $S$ be the set of agents that can afford their demand at $\vec{p}$. Consider the finest partition $\mathcal{P} = (I_1, \ldots, I_m)$ induced by the demand sets of the agents \emph{and} the allocation $\vec{x}$, such that every endpoint that appears in a demand set or in $\vec{x}$ is included as an end-point for some interval $I_j$. Then divide each interval $I_j$ into a number of pieces (of equal length) equal to the number of agents that want this piece. Let $\mathcal{P}'  = (J_1, \ldots, J_{r})$ be the resulting, finer partition, with these additional points resulting from cutting the intervals $I_j$. Now the problem can be interpreted as an instance with multiple divisible goods---one for every interval $J_{\ell}$---and valuations $v_{i,\ell}$ defined such that for any agent $i$ and good $J_{\ell}$, we have $v_{i,\ell} = 1$ if $J_{\ell} \in D_i$ and $v_{i, \ell} = 0$, otherwise. By Theorem \ref{thm:divisible_given_xp}, a supporting set of prices can be computed in polynomial time (if it exists), and these can be converted into a price curve in the cake cutting problem.

Given a price curve $\vec{p}$, we can compute again a partition induced by the demands of the agents, and solve the problem of finding the allocation by calling the algorithm for divisible goods in Theorem \ref{thm:divisible_given_xp} and casting back the solution to an allocation in the cake cutting problem.
\end{proof}

\begin{theorem}
Given an instance with single minded utilities and discrete goods, and an allocation $\vec{x}$, it can be decided in polynomial time if 
there exist equilibrium prices to support the allocation $\vec{x}$. However, it is co-NP-hard to determine if there exists an allocation $\vec{x}$ such that $(\vec{x}, \vec{p})$ is an equilibrium.
\end{theorem}
\begin{proof}
Given an allocation $\vec{x}$, it can be checked in polynomial time if the allocation is feasible. Afterwards, we can simply write a linear program with real variables $\epsilon$ and $p_1, \ldots, p_j$ to find the supporting prices (if any). Note that $\epsilon$ will correspond to the objective of the program and will have the property that for each agent $i$ that does not get its demand, $\vec{p}(\vec{x}_i) \geq 1 + \epsilon$.

For the other direction, of finding the allocation given the prices, this problem was shown to be co-NP-hard even when all the quantities are $1$ (see, e.g., \cite{BHM15}).
\end{proof}

\end{document}